\documentclass[11pt,a4paper,onecolumn,accepted=2021-12-16]{quantumarticle}
\pdfoutput=1
\usepackage{geometry}                
\usepackage{amsfonts}
\usepackage{bbold}
\usepackage{graphicx}
\usepackage{fullpage}
\usepackage{amssymb}
\usepackage[]{xcolor} 
\usepackage{epstopdf}
\usepackage{dsfont}
\usepackage{amsthm}
\usepackage{amsmath}
\usepackage{thmtools}
\usepackage{thm-restate}
\usepackage{hyperref}
\usepackage{bm}
\usepackage[normalem]{ulem}
\usepackage[natbib=true,style=alphabetic,backend=bibtex,bibencoding=ascii,isbn=false,url=false,useprefix=true,maxbibnames=6,minbibnames=6]{biblatex}
\usepackage{physics}
\makeatletter
\def\th@plain{%
  \thm@notefont{}
  \itshape 
}
\def\th@definition{%
  \thm@notefont{}
  \normalfont 
}
\makeatother
\usepackage{amsthm}

\newcommand{\be}{\begin{equation}} 							
\newcommand{\ee}{\end{equation}}
\newcommand{\ba}{\begin{align}}
\newcommand{\ea}{\end{align}}
\newcommand{\bematrix}{\left(\begin{matrix}}
\newcommand{\ematrix}{\end{matrix}\right)}

\theoremstyle{definition}

\theoremstyle{theorem}
\newtheorem{theorem}{Theorem}[section]
\newtheorem*{theorem*}{Theorem}

\theoremstyle{lemma}

\theoremstyle{proposition}
\newtheorem{proposition}[theorem]{Proposition}

\theoremstyle{corollary}

\theoremstyle{observation}

\theoremstyle{remark}

\title{Bounding the quantum capacity with flagged extensions }
\author{Farzad Kianvash}\thanks{farzad.kianvash@sns.it}
\affiliation{NEST, Scuola Normale Superiore and Istituto Nanoscienze-CNR, I-56126 Pisa, Italy}
\author{Marco Fanizza}
\affiliation{NEST, Scuola Normale Superiore and Istituto Nanoscienze-CNR, I-56126 Pisa, Italy}
\affiliation{F\'isica Te\`orica: Informaci\'o i Fen\`omens Qu\`antics, Departament de F\'isica, Universitat Aut\`onoma de 
Barcelona, 08193 Bellaterra (Barcelona) Spain}\thanks{marco.fanizza@sns.it}
\orcid{0000-0003-0802-8000}
\author{Vittorio Giovannetti} 
\affiliation{NEST, Scuola Normale Superiore and Istituto Nanoscienze-CNR, I-56126 Pisa, Italy}
\orcid{0000-0002-7636-9002}

\begin{document}
\maketitle
\begin{abstract}	
	In this article we consider flagged extensions of convex combination of quantum channels, and find general sufficient conditions for the degradability of the flagged extension. An immediate application is a bound on the quantum $Q$ and private $P$ capacities of any channel being a mixture of a unitary map and another channel, with the probability associated to the unitary component being larger than $1/2$. We then specialize our sufficient conditions to flagged Pauli channels, obtaining a family of upper bounds on quantum and private capacities of Pauli channels. In particular, we establish new state-of-the-art upper bounds on the quantum and private capacities of the depolarizing channel, BB84 channel and generalized amplitude damping channel. Moreover, the flagged construction can be naturally applied to tensor powers of channels with less restricting degradability conditions, suggesting that better upper bounds could be found by considering a larger number of channel uses. 
\end{abstract}

\section{Introduction}
Protecting quantum states against noise is a fundamental requirement for harnessing the power of quantum computers and technologies. In a transmission line or in a memory, noise is modeled as a quantum channel, and several accesses to the channel together with careful state preparation and decoding can protect quantum information against noise. The quantum capacity $Q$ of a channel is the maximal amount of qubits which can be transmitted reliably, per use of the channel. It can be expressed in terms of an entropic functional, the coherent information $I_c$, which can be computed as a maximization over quantum states used as inputs of the communication line. The quantum capacity~\cite{Lloyd1997,QCAP2,Devetak2005}  of a channel can then be obtained as a limit for large $n$ of $I_c$ per use of the channel, for $n$ uses of the channel. A striking feature of the problem is the potential super-additivity of the coherent information~\cite{shor1996quantum, DiVincenzo1998, Smith2007, Fern2008, Smith2008, Smith2011, Cubitt2015, Leditzky2018, Bausch2019, Siddhu2020a,Siddhu2020b,Siddhu2020c, Noh2020,Yu2020}, which  which hinders the direct evaluation of the quantum capacity imposing an infinite number of optimizations on Hilbert spaces of dimension that grows exponentially in $n$. The existence of an algorithmically feasible evaluation of the quantum capacity remains as one of the most important open problems in quantum Shannon theory~\cite{Holevo2019, Wilde2017}, while finding computable upper or lower bounds on the quantum capacity constitutes important progress. 

The phenomenon of superadditivity is not restricted to the quantum capacity, as it shows up also for the classical capacity~\cite{hastings2009}, the classical private capacity~\cite{Li2009} and the trade-off capacity region~\cite{Zhu2017,Zhu2019}. In this paper we are also interested in the classical private capacity, which is the optimal rate for classical communication protected by any eavesdropper~\cite{Cai2004,Devetak2005}. In general this capacity is larger than the quantum capacity, but the upper bounds we obtain in this paper hold for both capacities.

Extending previous results, in this work we formulate sufficient conditions to obtain non-trivial upper bounds on the quantum capacity, using the so-called flagged extensions. A flagged extension of a channel that can be written as convex combinations of other channels is such that the receiver gets, together with the output of one of the channels in the convex combination, a flag carrying the information about which of the channels acted. This technique is particularly effective for a class of channels of physical significance, the Pauli channels. A qubit Pauli channel describes random bit flip and phase flip errors, which is a fundamental noise model; moreover, any qubit channel can be mapped to a Pauli channel by a twirling map~\cite{Horodecki1998}, which does not increase the quantum capacity~\cite{Barnum1998}. With these new flagged extensions we improve the results of~\cite{Fanizza2020,Wang2019} for two important Pauli channels: the depolarizing channel and BB84 channel, which are both superadditive~\cite{shor1996quantum, DiVincenzo1998, Bausch2019}, and their quantum capacities are not known, despite a long history of efforts ~\cite{DiVincenzo1998,Adami1997,Bruss1998,Rains1999,Cerf2000,Rains2001, Smith2007, Fern2008, Smith2008a, Smith2008b,Ouyang2011,Sutter2017,Leditzky2018b,Leditzky2018c, Bausch2019}. We also find new bounds for quantum capacity of the generalized amplitude damping channel~\cite{Garcia-Patron2009,Bausch2018,Rosati2018,Khatri2020,Wang2019}, improving the results of~\cite{Wang2019}. The bounds we obtain are not necessarily the best bounds available with these techniques, being just good guesses among all the instances of flagged channels that satisfy the sufficient conditions. In fact, we obtain an infinite sequence of optimization problem depending on the number of uses of the channel, each of which gives a bound on the capacity. It is not clear if a phenomenon analogous to superadditivity appears in this scenario. Even with one use of the channel, different choices of Kraus operators give different bounds.

These bounds are based on the most fruitful technique to obtain upper bounds on the quantum capacity: finding a degradable extension of the channel (e.g.~\cite{Smith2008a, Smith2008b,Ouyang2011, Leditzky2018b}.) In fact, degradable channels~\cite{Devetak2005a, Cubitt2008} have the property that the coherent information is additive, therefore the quantum capacity is obtainable as the coherent information of the channel. Moreover, a fundamental property of capacities is that they are generically decreasing under composition of channels, a fact that has a clear operational justification. It is also known that the quantum capacity of a degradable channel is equal to its private capacity, therefore the quantum capacity of a flagged degradable extension of a channel is also an upper bound for the private capacity of the original channel. Moreover, when the channel is approximately degradable useful bounds can still be obtained~\cite{Sutter2017, Leditzky2018c}. 

In a previous paper~\cite{Fanizza2020} we contributed to this line of work by considering a flagged degradable extension of the depolarizing channel. While previous constructions~\cite{Wolf2007, Smith2008a, Smith2008b, Leditzky2018b} used orthogonal flags, our contribution was to consider non-orthogonal flags, showing that degradable extensions can be obtained even in this less restricted setting, obtaining better bounds. In a subsequent work~\cite{Wang2019} the author combined non-orthogonal flags with approximate degradability, improving the bounds further by searching for the flagged extension with the best bound from approximate degradability. However, from careful inspection it seems that the advantage of the approximate degradability technique in this context is that it finds exactly degradable extensions for a choice of flags that our analysis did not cover. In fact, in the present work we extend the sufficient conditions for degradability for flagged channels and we find even better bounds by exploiting richer flag structures, while being able to reproduce the bounds already obtained with approximate degradability techniques. {In principle, the new bounds could be obtained by extending the space of the flags, evaluating the bounds with the approximate degradability method applied to the flagged extension, and minimizing over the possible flags. Unfortunately, this brute force search with more flags and in a larger Hilbert space for the flags becomes rapidly unpractical.}

The outline of the paper is the following: after the preliminaries in Section \ref{pre}, we show the derivation of the sufficient conditions for degradability of flagged channels in Section \ref{suff}. We then apply this result to obtain a general bound on the quantum capacity of any channel which is the convex combination of a unitary channel and any other channel in {{Section~\ref{sec:genapp}.}} In Section \ref{Pauli}, { we briefly review qudit Pauli channels,} then rewrite our sufficient conditions for Pauli channels, where the bounds on the quantum capacity appear to have a simpler form; we also show that two explicit choices of degradable extensions give state-of-the-art bounds for the quantum and private capacities of the depolarizing channel and the BB84 channel. In Section~\ref{GADsec}, we apply the our method to bound the quantum and private capacity of the  generalized amplitude damping channel. In Section~\ref{Disc} we add some observations about the possibility of getting even better bounds with this method. {We conclude with a summary of the results.}
 
\section{Preliminaries}\label{pre}
We consider a finite dimensional Hilbert space $\mathcal{H}$, and we denote the space of linear operators on $\mathcal{H}$ as $\mathcal{L(H)}$. Then a quantum channel $\Lambda:\mathcal{L}(\mathcal{H}_A)\rightarrow\mathcal{L}(\mathcal{H}_B)$ is a Completely Positive Trace Preserving (CPTP) map with input system $A$ and output system $B$. Any CPTP map can be written in Kraus representation
\begin{equation}
\Lambda[\rho]=\sum_{i=1}^{r} K_i \rho K^\dagger_i\, ,
\end{equation}
for some collection of Kraus operators $\{K_i\}_{i=1,...,r}$ satisfying the normalization condition \\ $\sum_{i=1}^r K^\dagger_i K_i= I$. Equivalently, we can cast any quantum channel in the Stinespring representation
\begin{equation}
	\Lambda[\rho]=\Tr_E[V \rho V^\dagger]\, ,
\end{equation}
where $V$ is an isometry from $\mathcal{H}_A$ to $\mathcal{H}_B\otimes\mathcal{H}_E$, {{the system $E$ formally acting  as the channel environment.}} For any Stinespring dilation one can define a complementary channel of $\Lambda$, that is a channel $\tilde{\Lambda}:\mathcal{L}(\mathcal{H}_A)\rightarrow\mathcal{L}(\mathcal{H}_E)$ defined as 
\begin{equation}\label{compl}
\tilde{\Lambda}[\rho]=\Tr_B[V \rho V^\dagger]\, .
\end{equation}
Stinespring dilations and complementary channels are uniquely identified up to an isometry acting on the environment. If there exist a degrading channel $W$ such that $W\circ \Lambda =\tilde{\Lambda}$, we say that $\Lambda$ is degradable~\cite{Devetak2005a} {(since for a given channel different complementary channels are related by an isometry, the corresponding degrading maps are also related by the same isometry)}. { We call $\mathbb{\Lambda}$ a degradable extension of $\Lambda$ if $\mathbb{\Lambda}$ is degradable and there exist another channel $N$  such that  $N\circ\mathbb{\Lambda}=\Lambda$. This is a particular example of a degradable lifting \cite{Winter2016}.}

The quantum capacity of a channel $Q(\Lambda)$ is the maximum asymptotic rate at which quantum information can be transmitted reliably using quantum channel $\Lambda$. {{In~\cite{Lloyd1997,QCAP2,Devetak2005} it was showed that this quantity 
can be computed as 
\begin{equation}
	Q(\Lambda)=\lim_{n\rightarrow \infty}\frac{1}{n}Q_n(\Lambda)\;, \qquad  Q_n(\Lambda)=\max_\rho I_c(\Lambda^{\otimes n},\rho)\;, 
\end{equation}
where $I_c(\Lambda,\rho):=S(\Lambda(\rho))-S(\tilde{\Lambda}(\rho))$
is the coherent information functional and $S(\rho):=-\Tr[\rho\log\rho]$ is the von Neumann entropy.}} {Since the number of parameters grows exponentially with $n$, the maximization over all input states of $n$ uses of channel is computationally demanding. In addition, one also needs to solve the optimization problem for any $n$ in order to use the regularized formula to evaluate the capacity.}  
We also consider the private capacity of a channel $P(\Lambda)$~\cite{Cai2004,Devetak2005}, which is the maximum rate of classical communication protected from any eavesdropper. This capacity also has a regularized expression

\begin{equation}
	P(\Lambda)=\lim_{n\rightarrow \infty}\frac{1}{n}P_n(\Lambda)\, , \qquad 
	P_n(\Lambda)=\max_{{\cal E}_n}\{  \chi(\Lambda^{\otimes n} ({\cal E}_n)) - \chi(\tilde{\Lambda}^{\otimes n}({\cal E}_n))\} \;,
\end{equation}
where now the maximization has to be performed over all possible ensembles {${\cal E}_n:=\{p_i,\rho_i\}$ of input states for $n$ channel uses $\Lambda^{\otimes n}$, and Holevo information of $n$ channel uses is defined as 
	\begin{equation}
		 \chi(\Lambda^{\otimes n} ({\cal E}_n))=S(\Lambda^{\otimes n} (\sum_i p_i \rho_i))-\sum_i p_i S(\Lambda^{\otimes n} (\rho_i)).
	\end{equation}
}In general, $P_n(\Lambda)\geq Q_n(\Lambda)$ and therefore $P(\Lambda)\geq Q(\Lambda)$. Additivity of $P(\Lambda)$ and $Q(\Lambda)$ means

{\begin{equation}
	Q_n(\Lambda)=nQ_1(\Lambda)=nQ(\Lambda). \qquad P_n(\Lambda)=nP_1(\Lambda)=nP(\Lambda).
\end{equation}

In general additivity does not hold. However, for any degradable channel $\Lambda$, the coherent information is additive~\cite{Devetak2005a}, and we get a single letter formula for the quantum capacity. Moreover, for degradable channels~\cite{Smith2008a}:
{\begin{equation}
P(\Lambda)=Q(\Lambda)=Q_1(\Lambda),
\end{equation}} 
therefore any upper bound we find on the quantum capacity via a degradable extension is also an upper bound for the private capacity.}

\section{Sufficient conditions for degradability of flagged extensions}\label{suff}
We outline a systematic construction of degradable flagged extensions for any convex combination of channels, i.e. channels of the form $\Lambda=\sum_{i=0}^lp_i\Lambda_i$, {{with $\{ p_i\}_{i=0,\cdots l}$ a set of probabilities and with 
$\Lambda_i$ channels themselves.}} We establish the following:
\begin{proposition}[Sufficient conditions for degradability of flagged extensions]
{{Let $\Lambda$ be a channel acting on the quantum system $A$ and its flagged extension 
\begin{equation}
\mathbb{\Lambda}=\sum_{i=0}^lp_i \Lambda_i \otimes\ket{\phi_i}\bra{\phi_i}\, , \label{FLAGNEW} 
\end{equation}
with $|\phi_i\rangle$ normalized states of an auxiliary system $F$.
The map $\mathbb{\Lambda}$}}  is degradable if there exists a choice of Kraus operators{{ $\{K_j^{(i)}\}_{j=1,..., r_i}$}} for each channel $\Lambda_i$ {and }{{an orthonormal basis $\{ |i\rangle\}_i$ for the space of $F$, {such that}}} 
\begin{equation}\label{deg cond}
	\bra{i'}\ket{\phi_i}\sqrt{p_i}K^{(i')}_{j'} K^{(i)}_j=\bra{i}\ket{\phi_{i'}}\sqrt{p_{i'}}K^{(i)}_{j} K^{(i')}_{j'}\qquad \forall i,j,i',j' \, .
\end{equation}  
\end{proposition}
\begin{proof}

{{Observe that starting from  a Kraus set  $\{K_j^{(i)}\}_{j=1,..., r_i}$ of the channel $\Lambda_i$, we can construct 
the following  isometric Stinespring dilation for such channel,
\begin{equation}
V_i\ket{\psi}_A:=\sum_{j=1}^{r_i}K^{(i)}_j\ket{\psi}_A\ket{i}_{B}\ket{j}_{\bar{B}}, 
\end{equation}
for all $|\psi\rangle_A$ states of $A$, 
with the systems $B$ and $\bar{B}$ playing the role of the effective channel environment. 
A  Stinespring representation of the flagged channel~(\ref{FLAGNEW}) can then be obtained  as}} \begin{equation}
V\ket{\psi}_A:=\sum_{i=0}^{l}\sqrt{p_i} V_i\ket{\psi}_A\ket{\phi_i}_F\, ,
\end{equation}
{{which, via  Eq.~(\ref{compl}) allows to express the complementary}} of the flagged channel as
\begin{equation}
	\tilde{\mathbb{\Lambda}}[\ket{\psi}_A\bra{\psi}]=\sum_{i,j}\sqrt{p_i p_j}{_F\langle}{\phi_j}|{\phi_i}\rangle_F\Tr_{A}[V_i\ket{\psi}_A\bra{\psi}V^\dagger_j]\;. 
\end{equation}
 Our goal is to find a channel such as $W$ that degrades the flagged channel to its complementary channel i.e. $W\circ\mathbb{\Lambda}=\tilde{\mathbb{\Lambda}}$. A natural candidate for the Stinespring representation of the degrading channel is as follows
\begin{equation}
	V'\ket{\psi}_A\ket{i}_F:=V_i\ket{\psi}_A\, .
\end{equation}
Consider hence  the following state
\begin{align} \nonumber 
	V'V\ket{\psi}_A&=\sum_{i=0}^{l}\sqrt{p_i}V' V_i \ket{\psi}_A\ket{\phi_i}_F  =\sum_{i=0}^{l}\sum_{i'=0}^l \sqrt{p_i}\bra{i'}\ket{\phi_i}V_{i'}V_{i}\ket{\psi}_A\\
	&=\sum_{i=0}^{l}\sum_{i'=0}^l \sum_{j=1}^{r_i}\sum_{j'=1}^{r_{i'}}\bra{i'}\ket{\phi_i}\sqrt{p_i}K^{(i')}_{j'} K^{(i)}_j \ket{\psi}_A\ket{i}_{B}\ket{j}_{\bar{B}}\ket{i'}_{B'}\ket{j'}_{\bar{B'}}\, ,
\end{align}
{{where for ease of notation $\bra{i'}\ket{\phi_i}$ stands for ${_F\langle}{i'}|{\phi_i}\rangle_F$. By construction
the}} states of subsystem $B\bar{B}$ and $B'\bar{B'}$ are equal to $\tilde{\mathbb{\Lambda}}[\ket{\psi}\bra{\psi}]$ and $W\circ\mathbb{\Lambda}[\ket{\psi}\bra{\psi}]$ respectively.  Therefore, a sufficient condition for the degradability of $\mathbb{\Lambda}$ is that $V'V\ket{\psi}_A$ is invariant if we swap subsystem $B\bar{B}$ with $B'\bar{B'}$. {Writing the swap operator as $S_\leftrightarrow$, we have}
\begin{align}\label{swapped state}
	S_\leftrightarrow V'V\ket{\psi}_A &=\sum_{i=0}^{l}\sum_{i'=0}^l \sum_{j=1}^{r_i}\sum_{j'=1}^{r_{i'}}\bra{i'}\ket{\phi_i}\sqrt{p_i}K^{(i')}_{j'} K^{(i)}_j \ket{\psi}_A\ket{i}_{B}\ket{j}_{\bar{B}}\ket{i'}_{B'}\ket{j'}_{\bar{B'}}\nonumber\\
	&=\sum_{i=0}^{l}\sum_{i'=0}^l \sum_{j=1}^{r_i}\sum_{j'=1}^{r_{i'}}\bra{i'}\ket{\phi_i}\sqrt{p_i}K^{(i')}_{j'} K^{(i)}_j \ket{\psi}_A\ket{i'}_{B}\ket{j'}_{\bar{B}}\ket{i}_{B'}\ket{j}_{\bar{B'}}\, \nonumber\\
	&=\sum_{i=0}^{l}\sum_{i'=0}^l \sum_{j=1}^{r_i}\sum_{j'=1}^{r_{i'}}\bra{i}\ket{\phi_{i'}}\sqrt{p_{i'}}K^{(i)}_{j} K^{(i')}_{j'} \ket{\psi}_A\ket{i}_{B}\ket{j}_{\bar{B}}\ket{i'}_{B'}\ket{j'}_{\bar{B'}}\,,\nonumber\\
	&=V'V\ket{\psi}_A,
\end{align}
{  where we used Eq.~(\ref{deg cond}) in the second equality. Since~$S_\leftrightarrow V'V\ket{\psi}_A=V'V\ket{\psi}_A$, the flagged channel $\mathbb{\Lambda}$ is degradable. }\\
\end{proof}
 {In this way we reduced the degradability conditions to commutation conditions on the Kraus operators. In the case that all the Kraus operators commute, we choose all of the flags to be equal to $\ket{\phi_i}=\ket{\phi}=\sum_j \sqrt{p_j}\ket{j}$ to satisfy the degradability conditions, therefore the original channel itself is degradable. In this way we can recover the known fact that channels with commuting Kraus operators are degradable~\cite{Devetak2005a}. In the case that the Kraus operators of each $\Lambda_i$ commute with each other i.e. $K^{(i)}_jK^{(i)}_{j'}=K_{j'}^{(i)}K_j^{(i)} \quad \forall j,j',i$, we can choose orthogonal flags $\ket{\phi_i}=\ket{i}$ to construct a degradable extension. However, our proof is not sufficient to recover the known fact that orthogonal flagged convex combination of degradable channels are degradable. Nonetheless, by allowing in our construction more freedom in the choice of $V'$ we can recover this fact. In particular, one can use $V'\ket{\psi}_A\ket{i}_{F}=\sum_{j=1}^{r_i}D^{(i)}_j\ket{\psi}_A\ket{i}_{B'}\ket{j}_{\bar{B}'}$, where $D^{(i)}_j$ are the Kraus operators of the complementary of the degrading map of $\Lambda_i$. Therefore, allowing arbitrary $D^{(i)}_j$ and checking for equality of partial traces in the systems $B B'$ and $\bar B \bar B'$ one obtains more general sufficient conditions, covering both the case of Proposition 3.1 and orthogonal flagged convex combination of degradable channels.
 }

\section{General applications} \label{sec:genapp} 
\subsection{Convex combination of a unitary operation with an arbitrary channel}\label{convunit}
Consider a channel that is obtained as a convex combination of a unitary mapping induced by the unitary operator $U$ plus an extra CPTP term  $\Lambda_1$, i.e.
\begin{equation}
	\Lambda[\rho]=(1-p)U\rho U^\dagger +p\Lambda_1[\rho] =(1-p)U\rho U^\dagger +p\sum_{j=1}^{r}K_j \rho K_j^\dagger\, ,
\end{equation}
where $p \in [0,1]$ and the $K_i$ being a Kraus set of $\Lambda_1$. As the quantum and private capacities are invariant under unitary transformations, in what follow without loss of generality we shall set $U$ as the identity map $I$, redefining the rest accordingly if necessary.
Following the construction of the previous section we hence define the flagged extension (see Eq.~(\ref{FLAGNEW})) of $\Lambda$ as
\begin{equation}\label{flagged 1u}
	\mathbb{\Lambda}[\rho]=(1-p)\rho\otimes \ket{\phi_0}\bra{\phi_0}+p\Lambda_1[\rho]\otimes\ket{\phi_1}\bra{\phi_1}\, ,
\end{equation}
for which  the degradability conditions in Eq.~(\ref{deg cond}) becomes 
{{\begin{align}\label{flagstwo}
	\bra{1}\ket{\phi_0}\sqrt{1-p}=\bra{0}\ket{\phi_1}\sqrt{p}\;, \qquad 
	\bra{1}\ket{\phi_1}K_jK_{j'}=\bra{1}\ket{\phi_1}K_{j'}K_j\, .
\end{align}}}
Since $K_j$ operators do not need to commute, we set $\bra{1}\ket{\phi_1}=0$ and if {$p\leq 1/2$} we get the following solution
\begin{equation}\label{twoflag_deg_cond}
		\ket{\phi_1}=\ket{0},\qquad \ket{\phi_0}=\sqrt{\frac{1-2p}{1-p}}\ket{0}+\sqrt{\frac{p}{1-p}}\ket{1}\, .
\end{equation}
Surprisingly, without any assumption on the form of $\Lambda_1$ we found a regime for which the channel in Eq.~(\ref{flagged 1u}) is degradable {with non-orthogonal flags.} 
Therefore, we get the following upper bound
\begin{equation}\label{boundsimp}
	Q(\Lambda)\leq Q(\mathbb{\Lambda})=Q_1(\mathbb{\Lambda})\, .
\end{equation}
{Note that, in the same regime $p\leq 1/2$, one also has that the extension with orthogonal flags is degradable. Consider indeed a map acting as follows on product states (it can be trivially extended): $W[\rho\otimes\ketbra{\phi_0}{\phi_0}]=\frac{1-2p}{1-p}\ketbra{0}{0}\otimes\ketbra{\phi_0}{\phi_0}+\frac{p}{1-p} \tilde\Lambda_1[\rho]\otimes \ketbra{\phi_1}{\phi_1}$, $W[\rho\otimes\ketbra{\phi_1}{\phi_1}]=\ketbra{0}{0}\otimes\ketbra{\phi_0}{\phi_0}$. One can verify that this map is a valid degrading map. However, the extension with non-orthogonal flags has a lower quantum capacity and therefore gives a better upper bound. Note also that one can also consider the family of extensions

\begin{equation}\label{convex_c}
	\mathbb{\Lambda}_c[\rho]:=(1-c^2)(1-p)\rho\otimes \ket{\phi_0}\bra{\phi_0}+\left(c^2(1-p)\rho+p\Lambda_1[\rho]\right) \otimes\ket{\phi_1}\bra{\phi_1}\, ,
	\end{equation}
which according to Eq.~(\ref{twoflag_deg_cond}) are degradable for $|\braket{\phi_0}{\phi_1}|^2=\frac{1-2(p+c^2-pc^2)}{1-(p+c^2-pc^2)}$, for $0\leq c^2\leq \frac{1-2p}{2(1-p)}$. Each of these extensions gives an upper bound, and the best bound is found by minimization.
\begin{equation}\label{minc}
	Q(\Lambda)\leq \min_{0\leq c^2\leq \frac{1-2p}{2(1-p)}}Q_1(\mathbb{\Lambda}_c)\, .
\end{equation}
Putting $\Lambda_1[\rho]=I/d$, one recovers degradable flagged extension of depolarizing channel. The best previous upper bound for qubit depolarizing channel is given in~\cite{Wang2019}. In~\cite{Wang2019}, by fixing the structure of the flagged extension as in Eq.~(\ref{flagged 1u}), the author found the optimal upper bound for the quantum capacity of the qubit depolarizing channel using approximate degradability, and minimizing over flagged extensions. In fact, the best bound with this method is obtained for an exactly degradable extension, and the flags giving the optimal upper bound in~\cite{Wang2019} are the same that we find analytically in Eq.~(\ref{flagstwo}). On the other hand, considering the family of extensions Eq.~(\ref{convex_c}) we can also recover mixed state flags associated to the identity channel, by writing

\begin{equation}\label{mix1001}
	\mathbb{\Lambda}_c[\rho]=(1-p)\rho\otimes \left((1-c^2)\ket{\phi_0}\bra{\phi_0}+c^2\ket{\phi_1}\bra{\phi_1}\right)+p\Lambda_1[\rho]\otimes\ket{\phi_1}\bra{\phi_1}\, .
	\end{equation}
For $c^2=\frac{1-2p}{2(1-p)}$, $\braket{\phi_0}{\phi_1}=0$ and we recover the best bound from~\cite{Fanizza2020}, obtained with the same flag structure. The minimization over $c$ gives in general some improvement on this class of upper bounds.}

Another non-trivial construction is the following: if $K_j=\sqrt{p_j} U_j$ for some unitaries $U_j$, with $\sum_{i=1}^r p_j =1-p$, then we can consider the mixed unitary flagged channel, 

\begin{equation}
	\Lambda[\rho]=(1-p)\rho \otimes \ketbra{\phi_0}  +\sum_{j=1}^{r} p_j U_j \rho U_j^\dagger \otimes \ketbra{\phi_j}{\phi_j}\;,\end{equation}
and any flag choice such that 
\begin{align}
\braket{i}{\phi_j}=0 \,\,\,\mathrm{if}\, i \neq 0 \, \mathrm{and} \,j \neq 0  \, \mathrm{and} \, i\neq j\;,
\end{align}
gives a degradable extension. 
{The best upper bounds we obtain from degradable extensions with more than two flags, for the depolarizing channel and BB84 channel, have exactly this flag structure. However, the degradability conditions are more general than this and we give a more specialized treatment to these channels in the following sections.
\\
Finally,  as an extension of the argument presented in Eq.~(\ref{mix1001}), we remark that even general extensions with mixed flags can be considered in this framework, by changing the convex combination considered. For example, a rank two flag can be introduced by splitting a term $K\rho K^\dagger \otimes (q\ketbra{0}{0}+(1-q) \ketbra{1}{1})=\sqrt{q}K\rho \sqrt{q}K^\dagger \otimes \ketbra{0}{0} +\sqrt{1-q}K\rho \sqrt{1-q}K^\dagger \otimes \ketbra{1}{1} $, where we now flag a channel with new Kraus operators $\sqrt{q}K$ and $\sqrt{1-q}K$, each with a pure flag associated.}
\section{Pauli channels}\label{Pauli}
In this section we concentrate on an important subclass of mixed unitary channels, the Pauli channels, which describe random bit flip and phase flip errors in qubits and
their generalization to qudits models. For this class of channels the structure of degradable flagged extensions is quite rich and the upper bounds can be made more explicit.
The following treatment of generalized Pauli channels follows the phase-space description of finite dimensional quantum mechanics~\cite{Wootters1987, Appleby2005, Gross2006, Gross2008, DeBeaudrap2013, Gross2017}. 

\subsection{Qubit Pauli group}
We start by recalling the Pauli group of one qubit:
\begin{equation}
\mathcal P:=\{\pm I, \pm i I, \pm X, \pm iX,\pm Y, \pm iY,\pm Z,\pm iZ\},
\end{equation}
where 
$$I=\left(\begin{matrix}1&0\\0&1\end{matrix}\right),\qquad X=\left(\begin{matrix}0&1\\1&0\end{matrix}\right),\qquad Y=\left(\begin{matrix}0&-i\\i&0\end{matrix}\right),\qquad Z=\left(\begin{matrix}1&0\\0&-1\end{matrix}\right). $$
The Pauli group of $n$ qubits is obtained as the tensor product of $n$ copies of the Pauli group of one qubit
$\mathcal P^n:=\{\otimes_{j=1}^{n} \omega_j | \omega_j \in \mathcal P\}$.
For our purposes it suffices to consider $\mathsf P^n:=\mathcal P^n/C^n$, the quotient of the Pauli group with its center $C^n:=\{\pm I^{\otimes n}, \pm i I^{\otimes n}\}$. 
Each element of $\mathsf P^n$ can be identified by a pair of $n$ bit-strings $x=(q,p)$ according to the definition 
\begin{equation}
P_{(q,p)}:=i^{- p\cdot q}  \otimes_{j=1}^n Z^{p_j}X^{q_{j}}.
\end{equation}
It is then immediate to see that for any two $x=(q,p),y=(q',p')$ we have $P_x P_y=(-1)^{\langle x, y\rangle}P_y P_x$, where 
\begin{equation}
\langle x, y \rangle= p\cdot q' -q\cdot p' \mod 2\;,
\end{equation}
{{and that  $\Tr[P_{x}]=2\delta_{x,0}$.}}
\subsection{Qudit Pauli group}

The generalization of the Pauli group for one qudit is the group $\mathcal W_d$ generated by $\tau I$ ($\tau:=e^{\frac{ (d^2+1)\pi i}{d}}$), and the Weyl-Heisenberg operators $X,Z$ acting as
\begin{equation}
X\ket{j}=\ket{j+1}\,\,\, \mod d,  \qquad Z\ket{j}=e^{j\frac{2\pi i}{d}}\ket{j} \qquad\,\,j={0,...,d-1}.
\end{equation} 
For several qudits, likewise we set 
$\mathcal W^n_d:=\{\otimes_{j=1}^{n} \omega_j | \omega_j \in \mathcal W_d \}.$
The center of this group is still a set of multiples of the identity $C_d^n=\{\tau^j I^{\otimes n}:j= 0,...,D-1\}$, where $D=d$ if $d$ is odd and $D=2d$ if $d$ is even; we define $\mathsf W^n_d:=\mathcal W_d^n/C_d^n$. 
Each element of $\mathsf W^n_d$ can be identified by a pair of $n$ Dit-strings $x=(q,p)\in \mathbb Z^{2n}_{D}$ according to the definition 
\begin{equation}
W_{(q,p)}:=e^{-\frac{ (d^2+1)\pi i}{d}(p\cdot q)}  \otimes_{j=1}^n Z^{p_j}X^{q_j}.
\end{equation}
By close inspection it holds that
\begin{equation}\label{commut}
W_{x}W_{y}=e^{\frac{2\pi i }{d} \langle x,y\rangle}W_{y}W_{x},
\end{equation}
where now
\begin{equation}
\langle x, y \rangle= p\cdot q' -q\cdot p' \mod D.
\end{equation}
Moreover, for any $x,z\in\mathbb Z^{2n}_{D}$
we have 
\begin{equation}\label{modd}
W_{x+d z}=(-1)^{(d+1)\langle x,z\rangle}W_{x}, 
\end{equation}
and $\Tr[W_x]=d^n\delta_{x,0}$.

\subsection{Flagged Pauli channels}

Each element of $\mathsf W ^n_d$ is a unitary matrix, therefore it describes a reversible evolution of the system of $n$ qudits. Pauli channels are defined as convex combination of Pauli unitaries:
\begin{equation}\label{paulichan}
\Phi_{\bf w}[\rho]=\sum_{x\in \mathbb Z_{d}^{2n}}w_{x} W_{x}\rho W^\dagger_{x},
\end{equation}
where now it suffices to sum over $\mathbb Z^{2n}_{d}$ instead of $\mathbb Z^{2n}_{D}$ because of Eq.~(\ref{modd}), and they are parametrically described by the probability distribution ${\bf w}(x)= w_x$. 
For these channels, our sufficient conditions for degradability are less stringent than in general, because of the relations (\ref{commut}) occurring for any couple of Pauli unitaries.
To construct the flagged version of these channels, we note that the flags live, without loss of generality, in a Hilbert space $\mathcal H_F$ of dimension $d^{2n}$, with computational basis $\{\ket{x}\}_{x\in \mathbb Z_d^{2n}}$. We also consider the space $\mathcal H_C\otimes \mathcal H_F$, with $\mathcal H_C$ isomorphic to $\mathcal H _F$, and denote the partial trace with respect to $\mathcal H_C$ as $\Tr_C[\cdot]$.\\ 
{ Consider a flagged Pauli channel $\Phi_{\bf \Psi}$
\begin{equation}\label{flaggedpauli}
\Phi_{\Psi}[\rho]=\sum_{x\in \mathbb Z_{d}^{2n}}w_{x} W_{x}\rho W^\dagger_{x}\otimes \ketbra{\phi_x}, 
\end{equation}
where the label $\Psi$ determines $\Phi_{\Psi}$ through the definition of the state $\ket{\Psi}\in\mathcal H_C\otimes \mathcal H_F$ as
\begin{equation}
\ket{\Psi}=\sum_{x\in\mathbb Z_{d}^{2n}}\sqrt{w_x}\ket{x}_C\otimes \ket{\phi_x}_F=\sum_{x\in\mathbb Z_{d}^{2n},y\in\mathbb Z_{d}^{2n}}\sqrt{w_x}\braket{y}{\phi_x}\ket{x}_C\otimes \ket{y}_F\, .
\end{equation}

}

We define the projectors $\Pi_j$ on $\mathcal H_C\otimes \mathcal H_F$ projecting on $\mathrm{span}\{\ket{x}\ket{y}-e^{j\frac{2\pi i}{d}}\ket{y}\ket{x}:\langle x,y\rangle=j \,\mathrm{mod}\, d\}$. For a probability vector $\bf w$ over $\mathbb Z^{2n}_{d}$, we denote its Shannon entropy as $S({\bf w}):=-\sum_{x\in \mathbb Z^{2n}_{d}} w_x \log w_x$.
With these definitions, we are equipped to establish the following proposition:
\begin{proposition}[Upper bound on the quantum capacity of Pauli channels]\label{flaggedpaulibound}
{ Given a Pauli channel $\Phi_{\bf{w}}$, for any} $\ket{\Psi}\in \mathcal H_C\otimes \mathcal H_F$ satisfying
\begin{align}\label{constraints}
 \Tr[\Pi_{j} \ketbra{\Psi}]=0 \,\,\,\forall j\in\{0,...,d-1\}
 \qquad \Tr[(\ketbra{x}\otimes I) \ketbra{\Psi}]=w_x \qquad \forall x\in \mathbb Z_{d}^{2n}.
\end{align}
{ the quantum and private capacities of $\Phi_{\bf w}$ satisfy}
\begin{equation}\label{newboundpauli}
Q(\Phi_{\bf w}) { \leq P(\Phi_{\bf w})}\leq n \log d-S({\bf w})+S(\Tr_C[\ketbra{\Psi}]).
\end{equation}
In particular, the optimal upper bound is obtained by minimizing $S(\Tr_C[\ketbra{\Psi}])$ with the constraints~(\ref{constraints}).
\end{proposition}

\begin{proof}
Given a channel of the form (\ref{flaggedpauli}), consider a state on $\mathcal H_C\otimes \mathcal H_F$ of the form

{ Any state $\ket{\Psi}$ on $\mathcal H_C\otimes \mathcal H_F$ satisfying the second condition in Eq.~(\ref{constraints}) can be written as

\begin{equation}\label{constraint1201}
\ket{\Psi}=\sum_{x\in\mathbb Z_{d}^{2n}}\sqrt{w_x}\ket{x}_C\otimes \ket{\phi_x}_F=\sum_{x\in\mathbb Z_{d}^{2n},y\in\mathbb Z_{d}^{2n}}\sqrt{w_x}\braket{y}{\phi_x}\ket{x}_C\otimes \ket{y}_F,
\end{equation}

identifying a flagged extension $\Phi_{\Psi}$ of $\Phi_{\bf{w}}$.} 
Moreover, the degradability conditions for $\Phi_{\Psi}$ can be rewritten as
\begin{align}
 \Tr[\Pi_{j} \ketbra{\Psi}]=0 \,\,\,\forall j\in\{0,...,d-1\},
 \end{align}
therefore $\Phi_{\Psi}$ is degradable.
For flagged degradable Pauli channels, $Q_1=Q$ has a very simple form.  {By the covariance property of (flagged) Pauli channels, i.e. $\Phi_{\Psi}[W_x\rho W_x^\dagger]=(W_x \otimes I )\Phi_{\Psi}[\rho]( W_x^\dagger \otimes I )$ for all $W_x\in \mathcal W_d^n$. Moreover, we can also write the coherent information as $I_c(\Phi_{\bf \Psi},\rho)=S(\Phi_{\Psi}[\rho])-S(\Phi_{\Psi}\otimes \mathcal I [|\rho\rangle \rangle\langle\langle \rho|])$, for any purification $|\rho\rangle\rangle$ of $\rho$ and $\mathcal I$ identity channel, therefore using unitarily invariance of the von Neumann entropy and covariance we have $I_c(\Phi_{\bf \Psi},\rho)=I_c(\Phi_{\bf \Psi},W_x\rho W_x^\dagger)$. By concavity of coherent information for degradable channels~\cite{Yard2005}, we thus get
\begin{equation}
I_c(\Phi_{\bf \Psi},\rho)=\frac{1}{d^{2n}}\sum_{x\in\mathbb Z_{d}^{2n}} I_c(\Phi_{\bf \Psi},W_x\rho W^\dagger_x)\leq  I_c(\Phi_{\bf \Psi},\frac{1}{d^{2n}}\sum_{x\in\mathbb Z_{d}^{2n}}W_x\rho W^\dagger_x)=I_c(\Phi_{\bf \Psi},\frac{I}{d^{n}}).
\end{equation}}
 Therefore, the maximum of coherent information corresponds to the maximally mixed state, which is purified by the maximally entangled state $\ket{\Xi}=\frac{1}{d^{n/2}}\sum_{j=0,...,d-1}\ket{j}\otimes \ket{j}$. It holds that ${\bra{\Xi}W_x^{\dagger}W_y \otimes I \ket{\Xi}=\frac{1}{d^n}\Tr[W_x^{\dagger} W_y]}=\delta_{x,y}$, therefore
\begin{align}
Q_1(\Phi_{\bf \Psi})&=S\left(\Phi_{\bf \Psi}\left[\frac{I}{d^n}\right]\right)-S(\Phi_{\bf \Psi}\otimes I [\ketbra{\Xi}])=n \log d+S(\sum_{x\in\mathbb Z_{d}^{2n}} w_x \ketbra{\phi_x})-S({\bf w}) \nonumber \\
&=n \log d-S({\bf w})+S(\Tr_C[\ketbra{\Psi}]),
\end{align}
hence leading to Eq.~(\ref{newboundpauli}).
\end{proof}
Note that the minimization problem suggested by Proposition (\ref{flaggedpaulibound}) is non-convex, therefore it is hard to treat numerically. Its solution is also not unique in general. However, some useful upper bounds on the quantum capacity can be obtained, case by case, by simply minimizing over families of states which satisfy the constraints, but can be expressed in terms of a few parameters. In this way, we obtain state-of-the art results for the depolarizing channel and the BB84 channel.
As a side comment, the treatment in this section does not seem to cover the flag choice of~\cite{Fanizza2020,Wang2019} for the depolarizing channel. However, this is easily amended by splitting the Kraus operator proportional to the identity in Eq.~(\ref{paulichan}) into two Kraus operators with suitable probabilities, and assigning a different flag to each of them, respecting the sufficient conditions for degradability {(see the comments after Eq.~(\ref{boundsimp}))}.

\subsection{Depolarizing channel}\label{depo}
The depolarizing channel on one qudit is
\begin{equation}\label{depchannel}
\Lambda_p^d[\rho]:=(1-\frac{d^2-1}{d^2}p) \rho + \frac{p}{d^2}\sum_{x\in \mathbb Z_d^2 \setminus \{0\}}W_x\rho W^\dagger_x\;.
\end{equation}
The symmetries of this channel causes some potential redundancies in the states that achieve the optimal upper bound according to Proposition (\ref{flaggedpaulibound}). Consider the unitary operation $U_{\sigma}$ indexed by permutations $\sigma \in S_{d^2-1}$ which act by permuting the orthogonal set $\{\ket{x}\}_{x\in \mathbb Z_d^2 \setminus \{0\}}$ while leaving $\ket{0}$ invariant. Then, for any state $\ket{\Psi}$ satisfying the constraints, $U_{\sigma} \otimes U_{\sigma}\ket{\Psi}$ also satisfies the constraints, and it has the same entanglement entropy $S(\Tr_C[\ketbra{\Psi}])=S(\Tr_C[U_{\sigma} \otimes U_{\sigma}\ketbra{\Psi}(U_{\sigma} \otimes U_{\sigma})^\dagger])$. We cannot establish if the minimization problem has a unique solution, but if this was the case, then we could restrict the candidate states to those which are invariant under $U_{\sigma}\otimes U_\sigma$ for every $\sigma\in S_{d^2-1}$. We just take this observation as a suggestion for a guess, and we minimize $S(\Tr_C[\ketbra{\Psi}])$ on this restricted family of states. This is convenient because $S(\Tr_C[\ketbra{\Psi}])$ can be determined analytically and we can reduce the problem to a one-parameter minimization.
\begin{proposition}\label{Spsicomputation1201}
Any $\ket{\Psi}$ satisfying the constraint Eq.~(\ref{constraint1201}) for the map of Eq.~(\ref{depchannel}) and the condition $\ket{\Psi}=U_{\sigma} \otimes U_{\sigma}\ket{\Psi}$ can be parametrized with three complex variables $\alpha=\braket{0,0}{\Psi}$, $\beta=\braket{0,x}{\Psi}$ for $x\neq 0$, $\gamma=\braket{x,x}{\Psi}$ for $x\neq 0$. Accordingly we can write 
\begin{equation}\label{Spsi}
S(\Tr_C[\ketbra{\Psi}])=-(d^2-2)|\gamma|^2\log( |\gamma|^2)-v_+\log v_+-v_-\log v_- ,
\end{equation}
 with
\begin{equation}\label{eigen}
v_{\pm}=\frac{1}{2} (|\alpha|^2 + |\gamma|^2+ 2 |\beta|^2 (d^2-1) \pm
   \sqrt{(|\alpha|^2 - |\gamma|^2)^2 + 4 (d^2-1) |\beta|^2 |\alpha+\gamma^*|^2}).
\end{equation}
\end{proposition}

\begin{proof}

From the constraints we have that $\beta=\braket{0,x}{\Psi}=\braket{x,0}{\Psi}$, and from the action of a  permutation $U_{xy}$ that exchanges $x,y\neq0$  we have $\bra{0,x}U_{xy}\otimes U_{xy}\ket{\Psi}=\braket{0,y}{\Psi}$. From the constraints we have that $\braket{x,y}{\Psi}=e^{-\frac{2\pi i }{d} \langle x,y\rangle}\braket{y,x}{\Psi}$ for $x\neq y\,,\,\, x,y\neq  0$, then $\bra{{x,y}}{U_{xy}\otimes U_{xy}\ket{\Psi}}=\braket{y,x}{\Psi}=e^{-\frac{2\pi i }{d} \langle x,y\rangle}\braket{y,x}{\Psi}=0$. Also, $\bra{{x,x}}\ket{\Psi}=\bra{{x,x}}U_{xy}\ket{\Psi}=\bra{{y,y}}\ket{\Psi}=\gamma$ when $ x,y\neq  0$. This completes the parametrization. The eigenvalues of $\Tr_C[\ketbra{\Psi}]$ can be determined from the singular values of the matrix $M_{xy}$ of coefficients of $\ket{\Psi}=\sum_{x\in\mathbb Z_{d}^{2n},y\in\mathbb Z_{d}^{2n},} M_{xy} \ket{x}_C\otimes \ket{y}_F$. We have that the coefficients of $M^\dagger M$ are
\begin{align}
M^\dagger M_{0,0}=|\alpha|^2 + |\beta|^2 (d^2-1) \qquad M^\dagger M_{0,x}= \alpha \beta^*+\beta \gamma^* , \,\,x\neq 0 \\\qquad M^\dagger M_{x,y}=|\beta|^2,\,\,  x\neq y, \,x,y \neq 0 \qquad M^\dagger M_{x,x}=|\beta|^2+|\gamma|^2, \,\,x\neq 0\;.
\end{align}
Then $M^\dagger M- |\gamma|^2 I$ has rank $2$ and the nonzero eigenvalues can be determined by solving a quadratic equation. 

\end{proof}

\begin{proposition}\label{prop-one-param}
For $\ket{\Psi}$ satisfying $\ket{\Psi}=U_{\sigma} \otimes U_{\sigma}\ket{\Psi}$, the minimization of $S(\Tr_C[\ketbra{\Psi}])$ is a one-parameter minimization problem.
\end{proposition}

\begin{proof}
From the expression of $S(\Tr_C[\ketbra{\Psi}])$ in Eq. (\ref{Spsi}) and from Eq. (\ref{eigen}) it's evident that the result does not depend on the phases of $\alpha$, $\beta$ and $\gamma$ except for the term $|\alpha+\gamma^*|$, which should be maximized. This happens without loss of generality if $\alpha$ and $\gamma^*$ are real and positive. Then the two constraints $|\alpha|^2+(d^2-1)|\beta|^2=(1-\frac{d^2-1}{d^2}p)$ and $|\beta|^2+|\gamma|^2=\frac{p}{d^2}$ eliminate the remaining two parameters, with the constraint that $\gamma$ is such that $|\alpha|,|\beta|,|\gamma|<1$.
\end{proof}

Summarizing, from Proposition~\ref{flaggedpaulibound} and~\ref{Spsicomputation1201} we obtain 

\begin{align}\label{superbound1201}
Q(\Lambda_p^d) &\leq P(\Lambda_p^d)\leq  \log d-S({\bf w})+S(\Tr_C[\ketbra{\Psi}])\nonumber\\
&=  \log d-\left(1-\frac{d^2-1}{d^2}p\right)\log\left(1-\frac{d^2-1}{d^2}p\right)-\frac{p(d^2-1)}{d^2}\log \frac{p}{d^2}\nonumber\\ 
&-(d^2-2)|\gamma|^2\log(| \gamma|^2)-v_+(|\gamma|^2)\log v_+(|\gamma|^2)-v_-(|\gamma|^2)\log v_-(|\gamma|^2),
\end{align}

where the allowed values of $\gamma$ and the dependence on $\gamma$ of $v_+(|\gamma|^2)$ and $v_-(|\gamma|^2)$ are obtained as in the argument of Proposition~\ref{prop-one-param}, i.e.
\begin{align}\label{vfinal1001}
v_{\pm}(x)&=\frac{1}{2} \left(1 - (d^2-2) x\right)\pm \frac{1}{2}\Bigg\{\left(1 - 2 \frac{d^2-1}{d^2} p + (d^2-2) x\right)^2\nonumber\\&+4(d^2-1)\left(\frac{p}{d^2}-x\right) \left(x d^2+2 \sqrt{x} \sqrt{x  \left(d^2-1\right)-2 \frac{p(d^2-1)}{d^2}+1}-2\frac{p(d^2-1)}{d^2}+1\right)\Bigg\}^{1/2}.
\end{align}

The bound $Q_{\mathrm{fmin}}$ obtained from this one-parameter minimization can be combined with the no-cloning bound~\cite{Bruss1998,Cerf2000,Smith2008a,Ouyang2011}
\begin{equation}  \label{UP2} 
Q(\Lambda_p^d)\leq \big(1-\tfrac{2p(d+1)}{d}\big) \log{d}\: .
\end{equation}
using the fact that the convex hull of upper bounds from degradable extensions of the depolarizing channel is itself an upper bound~\cite{Smith2008a,Ouyang2011}.
A comparison between the most competitive upper bounds for $d=2$ is shown in Figure \ref{depo2}, where we can see that the bound we obtained outperforms all previous bounds in the whole parameter region. An improvement with respect to previous bounds can be obtained also for generic $d$, and we show as an example the bound for $d=4$ in Figure \ref{depo4}. {In this latter case, the bound from the convex hull is improved considering also the bound from Eq.~(\ref{minc}).}

\begin{figure}\center
\includegraphics[width=0.8\linewidth]{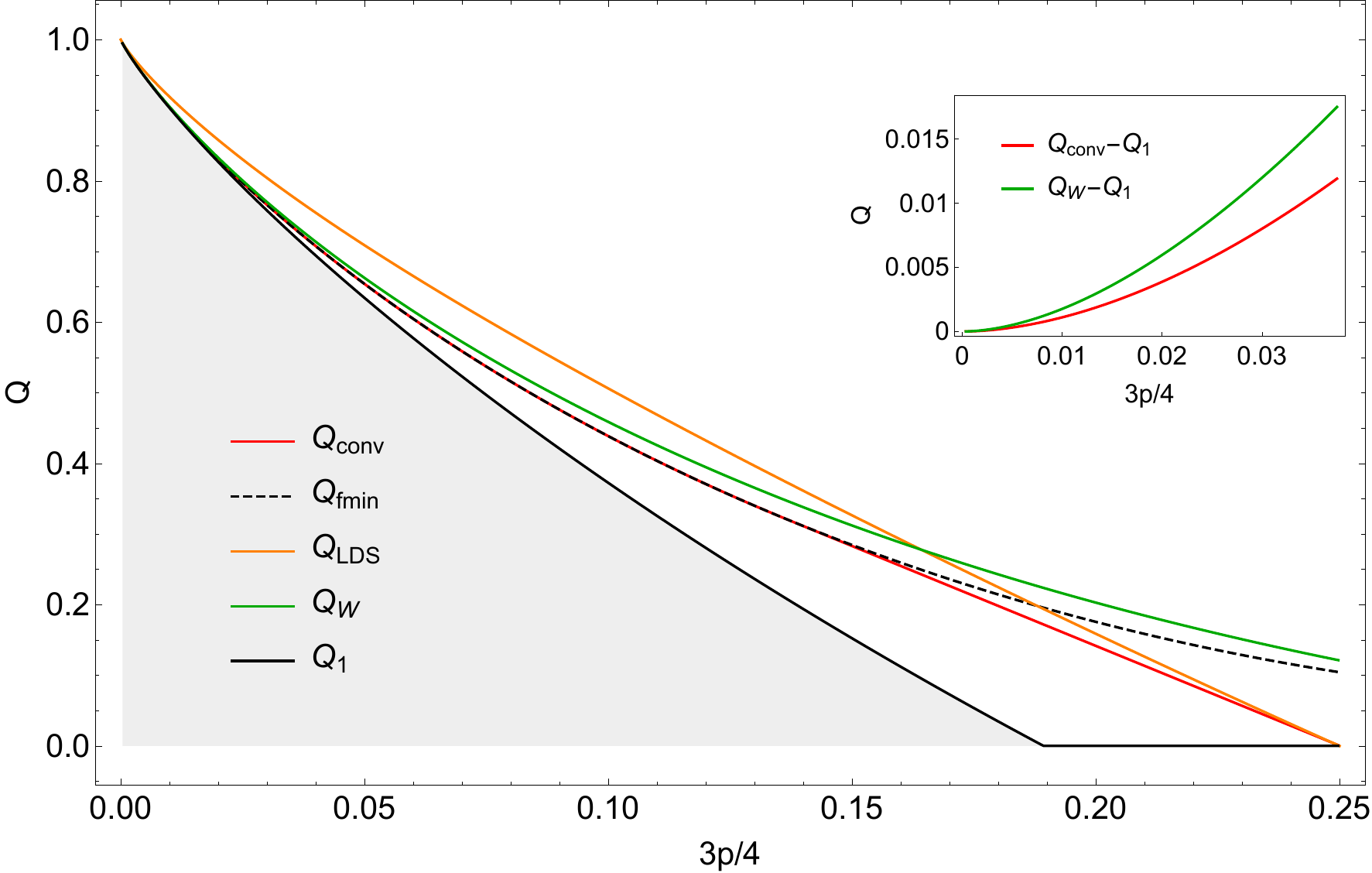}
\caption{Bounds on the quantum capacity of the depolarizing channel for $d=2$. Here $Q_{\mathrm {conv}}$ is the convex hull of the available upper bounds from degradable extensions, $Q_{\mathrm{fmin}}$ is the new upper bound, {obtained from Eq.~(\ref{newboundpauli}) by plugging in the expression Eq.~(\ref{Spsi}) and minimizing over $\gamma$, eliminating the other parameters as explained in the proof of Proposition \ref{prop-one-param}}. $Q_1$ is the lower bound given by the coherent information of one use of the channel. $Q_{LDS}$ is the bound from~\cite{Leditzky2018b} and $Q_W$ is the bound from~\cite{Wang2019}.}\label{depo2}
\end{figure}

\begin{figure}\center
\includegraphics[width=0.8\linewidth]{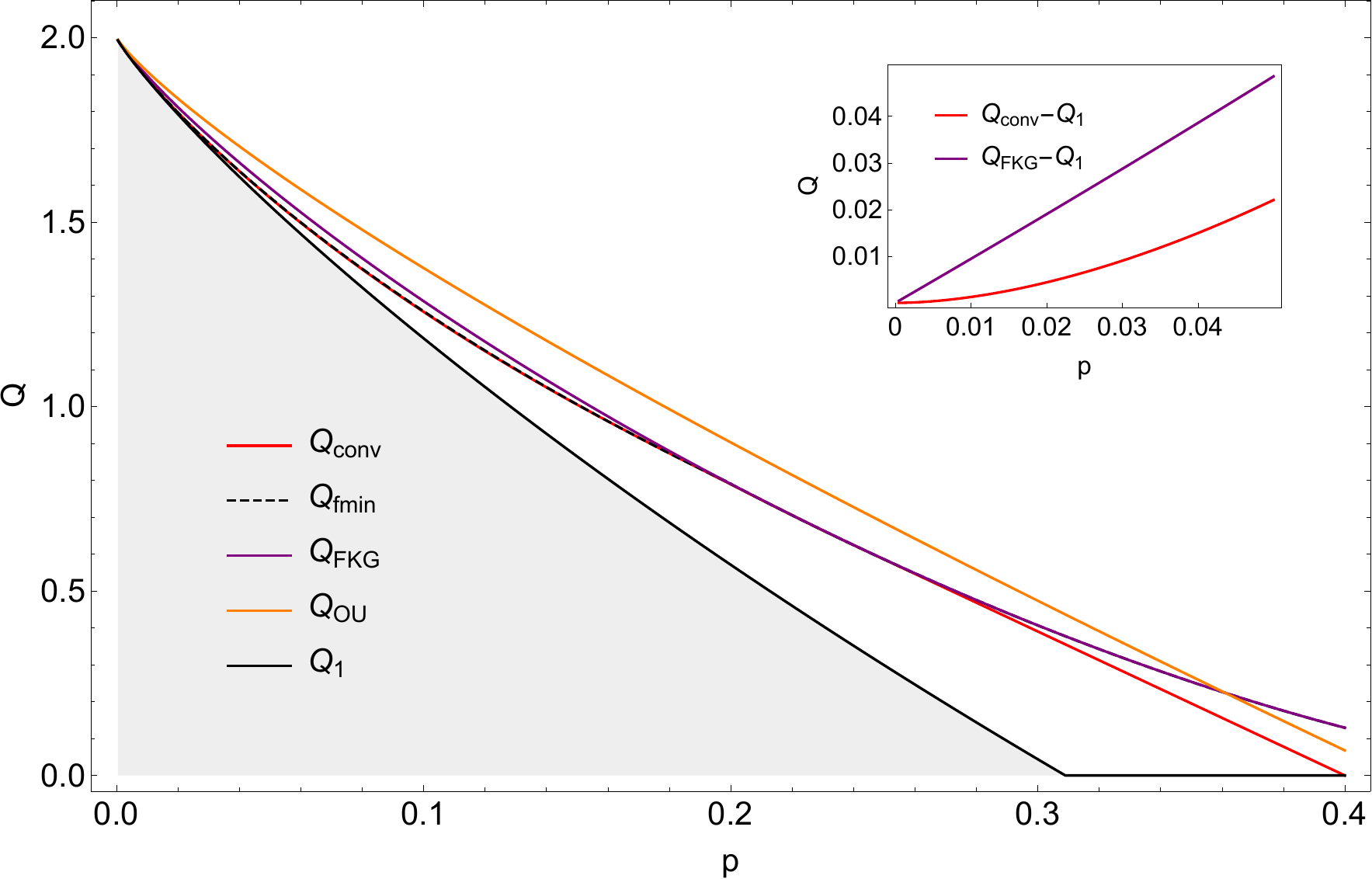}
\caption{Bounds on the quantum capacity of the depolarizing channel for $d=4$. Here $Q_{\mathrm {conv}}$ is the convex hull of the available upper bounds from degradable extensions, $Q_{\mathrm{fmin}}$ is the new upper bound, and $Q_1$ is the lower bound given by the coherent information of one use of the channel. {$Q_{FKG}$ is the bound from~\cite{Fanizza2020} and $Q_{\mathrm OU}$ is the bound from~\cite{Ouyang2011}. Note that in the main plot $Q_{\mathrm{fmin}}$ is the bound in Eq.~(\ref{minc}), since at scale used the bound with more flags is not noticeably better; the situation is different for very small $p$, in the regime plotted in the inset: {here we report $Q_{\mathrm {conv}}$ obtained from Eq.~(\ref{newboundpauli}) by plugging in the expression Eq.~(\ref{Spsi}) and minimizing over $\gamma$, eliminating the other parameters as explained in the proof of Proposition \ref{prop-one-param}} (see Eq.~(\ref{superbound1201}) and Eq.~(\ref{vfinal1001})).}}\label{depo4}
\end{figure}

\subsection{BB84 channel}

In this section we consider the channel that describes the famous quantum key distribution protocol by Bennett and Brassard~\cite{bennet14}. In its general form the channel is

\begin{align}\label{BB84sym}
B_{p_X,p_Z}[\rho] = (1-p_X-p_Z+p_Xp_Z)\rho+(p_X-p_Xp_Z)X\rho X+(p_Z-p_Zp_X)Z\rho Z+p_Xp_ZY\rho Y,
\end{align}

As in~\cite{Sutter2017}  and~\cite{Wang2019} we restrict to the case $p_X=p_Z=p$. The flagged extension we consider is

\begin{align}
B_{p, \Psi}[\rho] &= (1-p)^2\rho \otimes \ketbra{\phi_0}{\phi_0} +p(1-p)X\rho X \otimes  \ketbra{\phi_1}{\phi_1} +p(1-p)Z\rho Z \otimes \ketbra{\phi_2}{\phi_2}\nonumber\\&+p^2 Y\rho Y\otimes  \ketbra{\phi_3}{\phi_3}\, .
\end{align}
We choose the following parametrization for the flags
\begin{align}\label{flags BB84}
\ket{\phi_0}&=\sqrt{1-2\alpha^2-\beta^2}\ket{0}+\alpha\ket{1}+\alpha\ket{2}+\beta\ket{3}\nonumber\\
\ket{\phi_1}&=a\ket{0}+\sqrt{1-a^2-\gamma^2}\ket{1}-\gamma\ket{3}\nonumber\\
\ket{\phi_2}&=a\ket{0}+\sqrt{1-a^2-\gamma^2}\ket{2}-\gamma\ket{3}\nonumber\\
\ket{\phi_3}&=b\ket{0}+c\ket{1}+c\ket{2}+\sqrt{1-b^2-2c^2}\ket{3}\, ,
\end{align}
where we impose the parameters $\alpha,\beta,\gamma,a,b,c$ to be real and satisfying the normalization conditions for the vectors in Eq.~(\ref{flags BB84}). The degradability conditions in Eq.~(\ref{deg cond}) imply that $\alpha=a\sqrt{\frac{p(1-p)}{(1-p)^2}}$, $\beta=\frac{bp}{1-p}$ and $\gamma=c\sqrt{\frac{p}{1-p}}$. This is not the most general parametrization for the flags, however, because of the symmetry between the bit flip and phase flip error in Eq.~(\ref{BB84sym}), we chose this parametrization. Any set of flags in the form of Eq.~(\ref{flags BB84}) will result in a degradable extension of BB84 channel. Therefore, to get the best upper bound for the quantum capacity or private capacity of BB84 we should minimize the coherent information of its flagged channel with respect to three free parameters $a,b,c$. We have compared the result of the optimization with the previous bounds in Figure~\ref{fig:bb84}. The bound in~\cite{Wang2019} by Wang can be reproduced in our framework just by choosing $a=b=1,\,c=0$.

\begin{figure}
	\centering
	\includegraphics[width=0.8\linewidth]{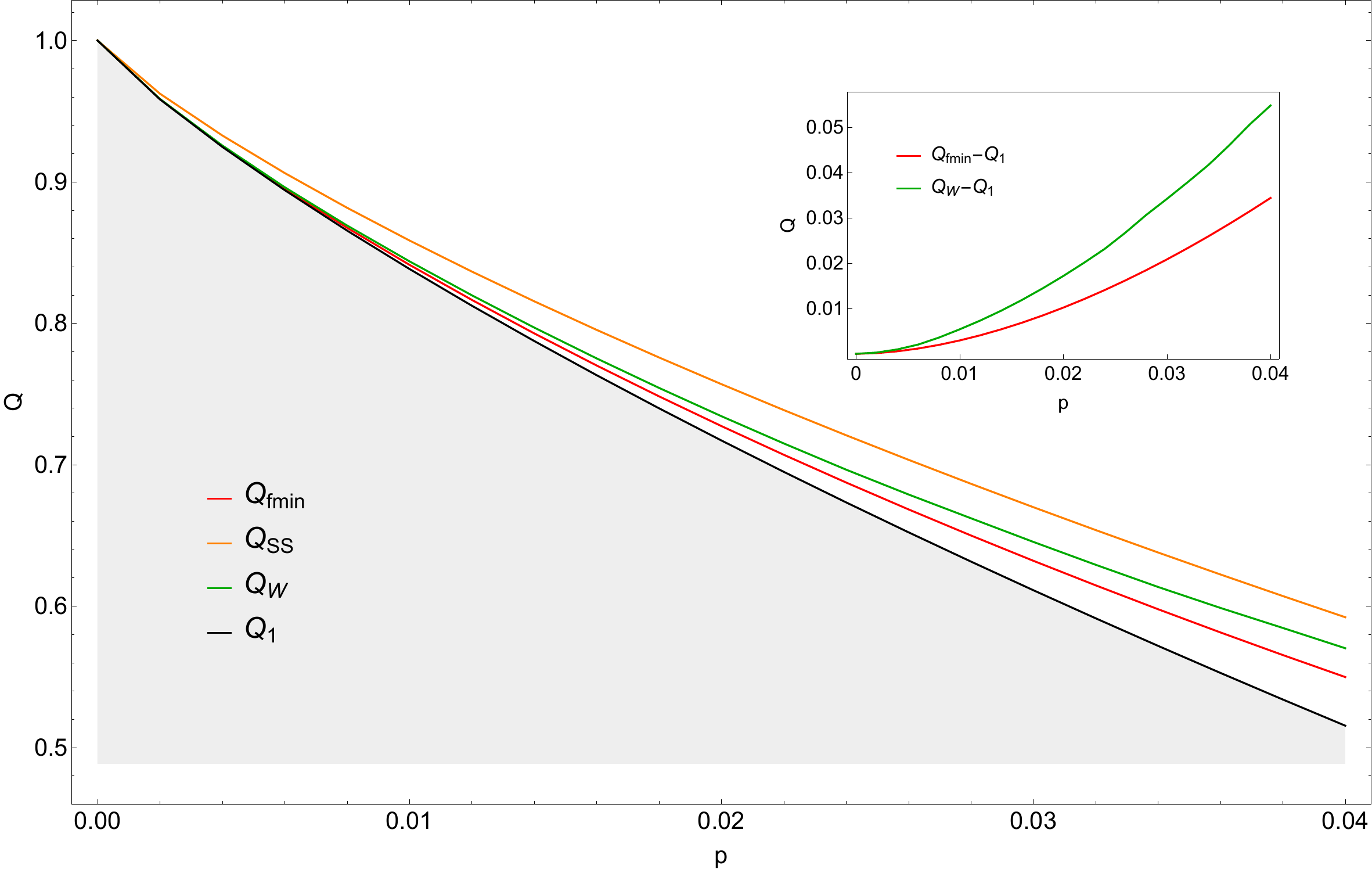}
	\caption{Bounds on the quantum and private capacity of BB84 channel. $Q_1$ is the coherent information of BB84 channel. $Q_{\mathrm{fmin}}$ is the new upper bound obtained by the degradable extension, {from Eq.~(\ref{newboundpauli}), using the parametrization for the flags in Eq.~(\ref{flags BB84}), for a suitable choice of the parameters}. $Q_W$ is the upper bound obtained in~\cite{Wang2019}. $Q_{SS}$ is the upper bound derived in~\cite{Smith2008a}.}
	\label{fig:bb84}
\end{figure}

\section{Generalized amplitude damping}\label{GADsec}
In this section we consider a bound on the quantum capacity of the generalized amplitude damping channel, which is a model of thermal loss on a qubit, relevant for quantum superconducting processors~\cite{Chirolli2008}. The generalized amplitude damping channel can be written as

\begin{equation}
\mathcal A_{y,N}[\rho]=N \,\mathcal A_{y}[\rho]+(1-N)\, X\circ \mathcal A_{y}\circ X [\rho ],
\end{equation}
where $\mathcal A_{y,N}$ is the conventional amplitude damping channel, with Kraus operators $K_1=(\ketbra{0}{0}+\sqrt{1-y}\ketbra{1}{1})$ and $K_2=\sqrt{y}\ketbra{1}{0}$. While $\mathcal A_{y}$, and $X\circ \mathcal A_{y}\circ X$ are degradable and their quantum capacity can be computed~\cite{Giovannetti2005}, their convex combination is not, and its quantum capacity is not determined. Previous upper bounds have been obtained by~\cite{Rosati2018,Khatri2020,Garcia-Patron2009,Wang2019}. In particular~\cite{Wang2019} used the following flagged extension together with approximate degradability to get the tightest bound available:

\begin{equation}
\mathcal A_{y,N}^F[\rho]=N \,\mathcal A_{y}[\rho]\otimes \ketbra{0}{0}+(1-N)\, X\circ \mathcal A_{y}\circ X [\rho ]\otimes \ketbra{1}{1},
\end{equation} 

In fact this extension is exactly degradable: the output of a complementary channel is 
\begin{equation}
\tilde{\mathcal A}_{y,N}^F[\rho]=N \,\tilde{\mathcal A}_{y}[\rho]\otimes \ketbra{0}{0}+(1-N)\, \tilde{X\circ \mathcal A_{y}\circ X} [\rho ]\otimes \ketbra{1}{1},
\end{equation}
and if the degrading map of $A_y$ is $W_y$, we have $\tilde{\mathcal A}_{y,N}^F[\rho]=(W_y\otimes \ketbra{0}{0}+ X\circ W_y\circ X \otimes \ketbra{1}{1})\circ \mathcal A_{y,N}^F[\rho]$. The quantum capacity of this extension can be evaluated to be $Q({\mathcal A}_{y,N}^F)\leq (1-N)I_c(A_y,\rho)+ N I_c(X\circ A_y \circ X,\rho)=Q(A_y)$. This simple bound seem to not have been pointed out previously, and the actual quantum capacity $Q({\mathcal A}_{y,N}^F)$ is very close to it.
Moreover, the structure of the generalized amplitude damping is such that one can get better bounds from different degradable extensions, where we adapt the argument by~\cite{Smith2008a} on the depolarizing channel:
\begin{proposition}[Combining bounds of degradable extensions of generalized amplitude damping]
For any collection of degradable extensions $\mathcal A_{y,N}^{ext,i}$, $i=1,...,l$, for any $y_0$ the quantum capacity of $\mathcal A_{y_0,N}$ is upper bounded by the convex hull of $Q(\mathcal A_{y_0,N}^{ext,i})$, $i=1,...,l$, as functions of the variable $N$.
 \end{proposition}
\begin{proof}
 For any $N_1, N_2$ such that $N=q N_1+ (1-q) N_2$, $1-N=1-qN_1+(1-q)N_2=q(1-N_1)+(1-q)(1-N_2)$, we have
\begin{equation}
\mathcal A_{y,N}[\rho]=q(N_1\,\mathcal A_{y}[\rho]+(1-N_1)\, X\circ \mathcal A_{y}\circ X [\rho ])+ (1-q)(N_2 \,\mathcal A_{y}[\rho]+(1-N_2)\, X\circ \mathcal A_{y}\circ X [\rho ])
\end{equation}

If $\mathcal A_{y,N_1}^{ext,i}$ and $\mathcal A_{y,N_2}^{ext,j}$ are degradable extensions of $\mathcal A_{y,N_1}$ and $\mathcal A_{y,N_2}$ respectively, then $q\mathcal A_{y,N_1}^{ext,i}\otimes\ketbra{0}{0}+(1-q)\mathcal A_{y,N_2}^{ext,j}\otimes \ketbra{1}{1}$ is a degradable extension of $\mathcal A_{y,N}$ with quantum capacity less than $qQ(\mathcal A_{y,N_1}^{ext,i})+(1-q)Q(\mathcal A_{y,N_2}^{ext,j})$.

\end{proof}

In addition to the extension proposed by~\cite{Wang2019}, we find two other degradable extensions using the results of this paper. The first is obtained observing that the following set is also a good choice of Kraus operators
\begin{align}
A_1&=\sqrt{N(1-N)}(\sqrt{1-y}+1)(\ketbra{0}{0}+\ketbra{1}{1})=\sqrt{N(1-N)}(\sqrt{1-y}+1)I\\
A_2&=\sqrt{(1-N)y}\ketbra{1}{0}\\
A_3&=((1-N)-N\sqrt{1-y})\ketbra{0}{0}+((1-N)\sqrt{1-y}-N)\ketbra{1}{1}\\
A_4&=\sqrt{Ny}\ketbra{0}{1}
\end{align}
We notice that $A_1$ is a rescaled unitary operator, therefore we can directly apply the bound of Eq.~(\ref{boundsimp}) with $(1-p)=N(1-N)(\sqrt{1-y}+1)^2$. This bound is applicable if $N(1-N)(\sqrt{1-y}+1)^2>1/2$. Moreover, at $N=1/2$, the generalized amplitude damping becomes a Pauli channel:
\begin{align}
\mathcal A_{y,0.5}(\rho)&=\frac{1-y/2+\sqrt{1-y}}{2} \rho+ \frac{1-y/2-\sqrt{1-y}}{2} Z \rho Z+ \frac{y}{4}( Y \rho Y+X \rho X)
\end{align}
and we get a more refined bound $Q_{\mathrm{fmin}}(y)$, using the techniques of the previous sections, in particular with the same flag structure of BB84 Eq.~(\ref{flags BB84}).
Putting all together, we observe that the bound by~\cite{Wang2019} remains the best one at high $y$, but at low $y$ it is beaten by the following bound allowed by the convex hull argument:
\begin{equation}\label{ampconvhull}
Q_{\mathrm{conv}}(y,N)=2NQ_{\mathrm{fmin}}(y)+(1-2N)Q(\mathcal A_{y})
\end{equation}
and using the full convex hull bound does not give substantial improvements. We plot the results in Figure~\ref{gadfig}.
\begin{figure}
	\centering
	\includegraphics[width=0.55\linewidth]{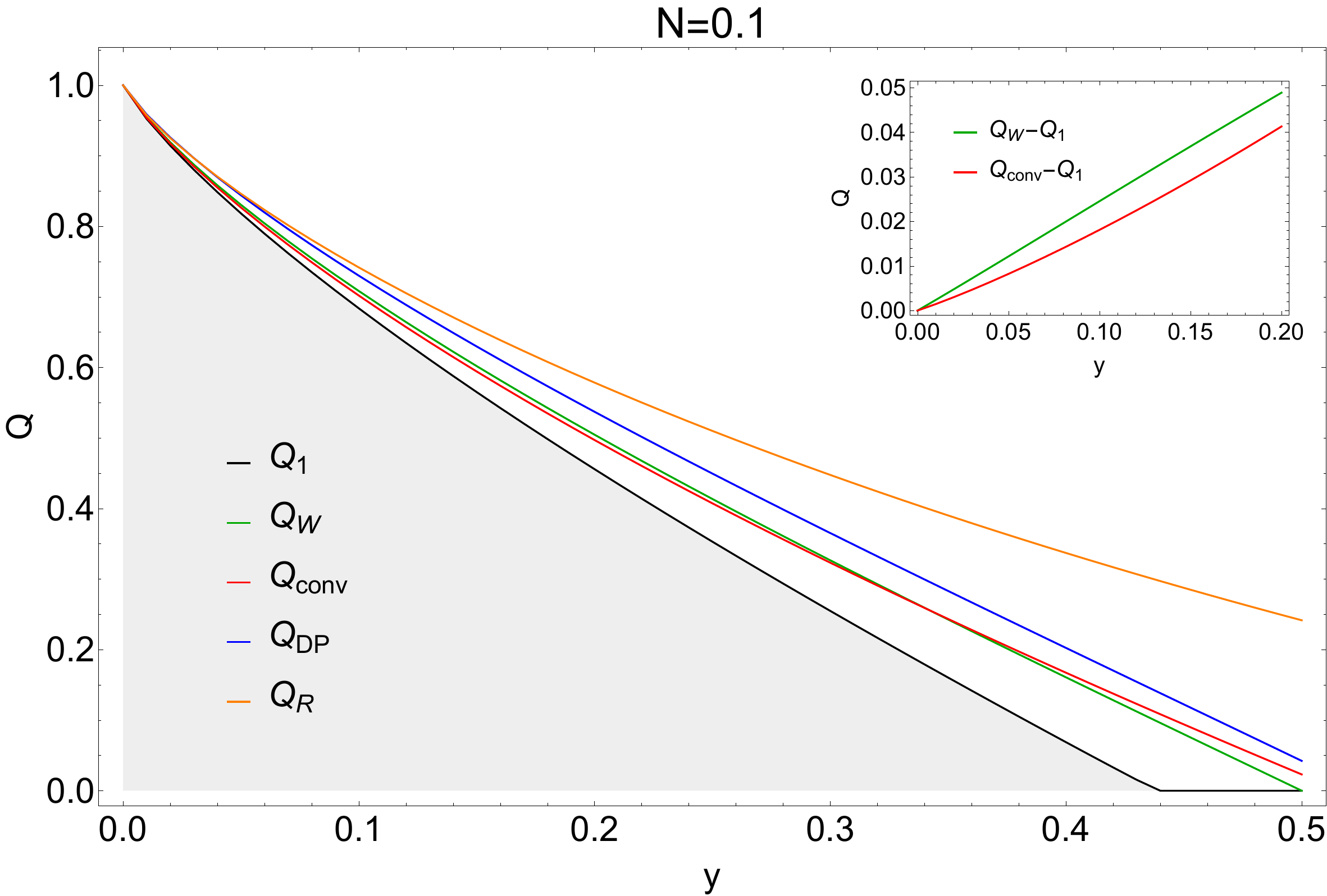}
	\vskip 0.5 cm
	\includegraphics[width=0.55\linewidth]{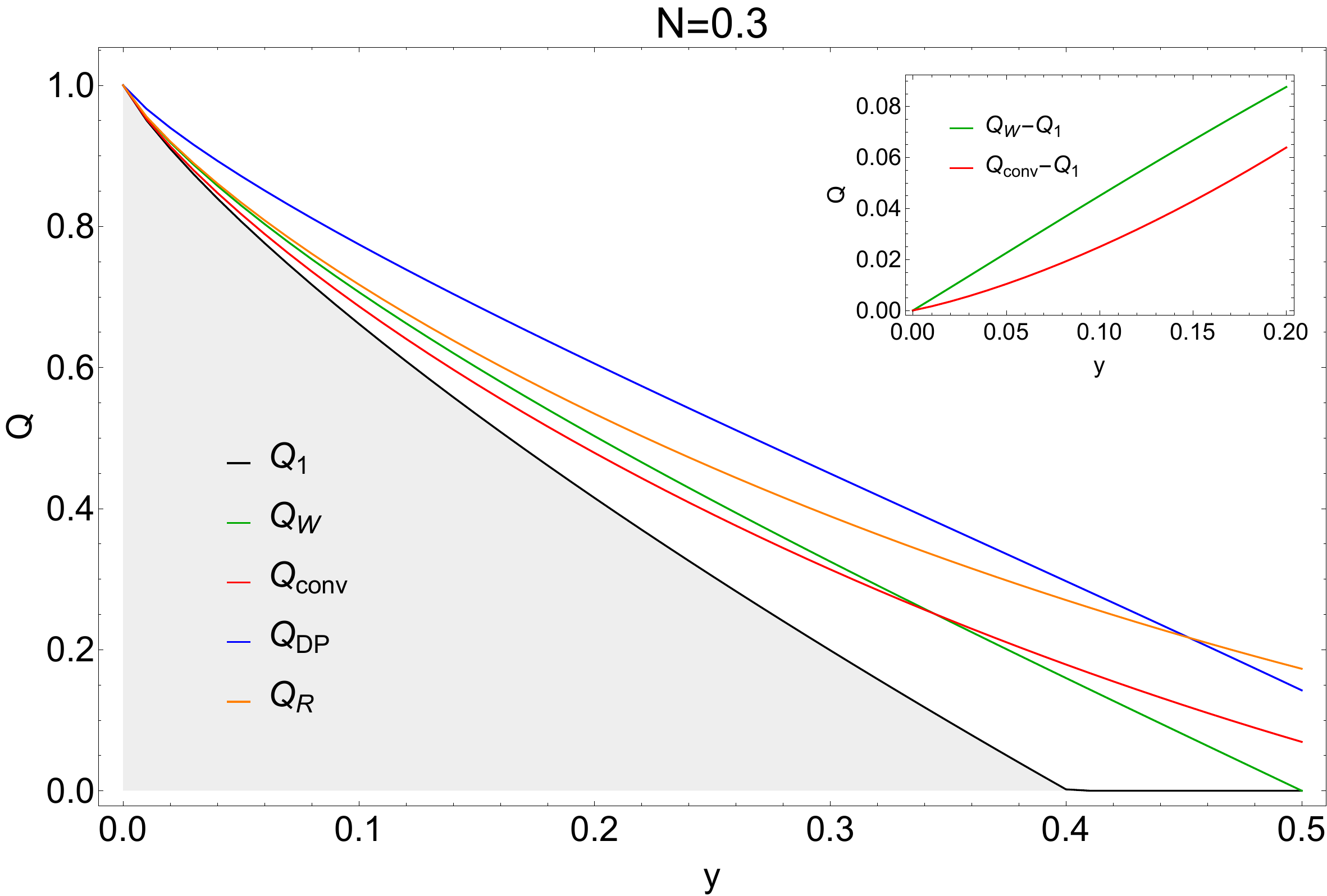}
	\vskip 0.5 cm
	\includegraphics[width=0.55\linewidth]{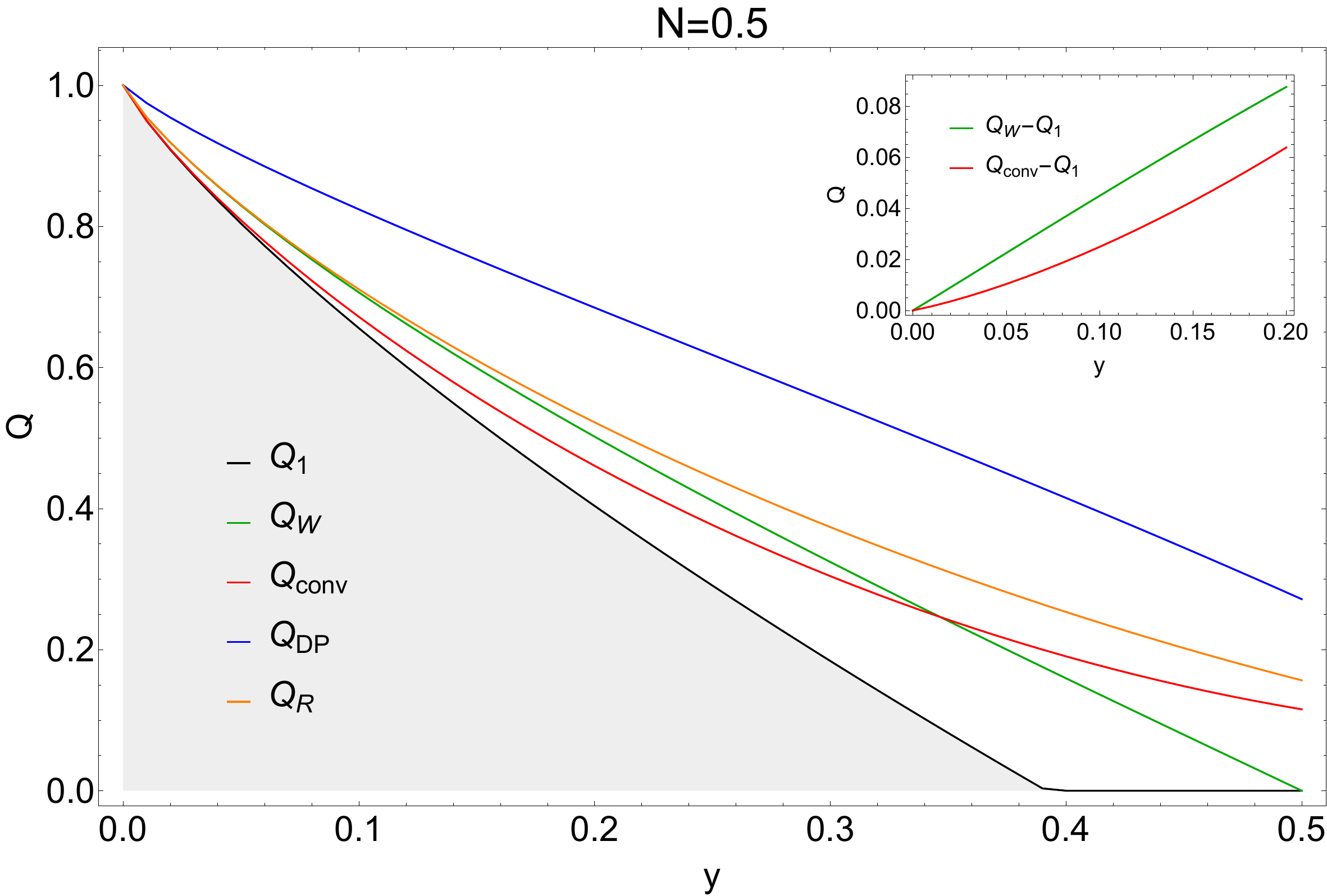}
	\caption{Bounds on the quantum capacity of the generalized amplitude damping, for three values of $N$. $Q_1$ is the lower bound given by the coherent information of one use of the generalized amplitude damping (see also lower bounds from~\cite{Bausch2018}, which give improvements mildly visible at this scale). $Q_{\mathrm{conv}}$ is the new upper bound {from Eq.~(\ref{ampconvhull})}, $Q_W$ is the upper bound obtained by Wang~\cite{Wang2019}, $Q_{DP}$ and $Q_R$ are obtained in~\cite{Khatri2020} respectively from data-processing and Rains information. Previous upper bounds~\cite{Rosati2018,Garcia-Patron2009} are worse and not plotted. }
	\label{gadfig}
\end{figure}

\section{Discussion}\label{Disc}

In the examples we provided we did not try to numerically optimize in the whole parameter region allowed by the sufficient conditions for degradability. Indeed, the minimization of the upper bound is not a convex optimization problem and would require brute force search, but there are already many parameters for Pauli channels and $d=2$, $n=1$. However, we stress the fact that the family of upper bounds for the quantum and private capacity of a channel $\Lambda$ is even larger in principle, as one can consider the flagged extension of $\Lambda^{\otimes k}$: a degradable flagged extension of $\Lambda$ gives also a degradable flagged extension of $\Lambda^{\otimes k}$ but the converse is not true. We tried to search for a better upper bound of the quantum capacity of the depolarizing channel with two uses, by restricting brute force search to certain parametrizations of the flags, but the attempts we made did not show anything better than the $k=1$ bound. It is desirable to further investigate if $k=1$ gives already the best upper bound or a phenomenon similar to superadditivity shows up for flagged extensions.
Moreover, the extensions we obtained are explicitly dependent on the Kraus representation chosen, which is not unique. We did not find a way to identify an optimal choice of Kraus operators. Since for any channel the number of Kraus operators can be increased arbitrarily by a suitable isometry, it is conceivable that one needs to look at an unbounded space of flags to optimize the upper bound of degradable extensions. As a side note, we point out that the mixed flags extensions as considered in~\cite{Fanizza2020} can be treated in the formalism of this paper, just by splitting the Kraus operators into Kraus operators with proper probabilities. 

\section{Conclusions}

We have introduced a method to construct degradable extension of quantum channels which can be written as the convex sum of other channels. This method is of general applicability, and we showed that it gives state-of-the-art upper bounds on the quantum and private capacity of two important Pauli channels, the depolarizing channel and the BB84 channel, and of the generalized amplitude damping channel. By virtue of its simplicity, we believe it can be used with success for many other channels too. The method could in principle give better bounds by considering different Kraus representations and flagged extensions of several uses of a channel, and it would be interesting to study its limitations.

\textit{Note added}: after the initial submission of this paper, the authors have also obtained upper bounds for the quantum and private capacities of phase-insensitive Gaussian channels, using an extension of the technique proposed here~\cite{fanizza2021estimating}.
 \\
\section*{Acknowledgments}
{{We thank Felix Leditzky, Matteo Rosati and Xin Wang for valuable comments. The authors acknowledge support by MIUR via PRIN 2017 (Progetto di Ricerca di Interesse Nazionale): project QUSHIP (2017SRNBRK).}}
\sloppy
\printbibliography
\end{document}